\pdfoutput=1
\RequirePackage{ifpdf}
\ifpdf 
\documentclass[pdftex]{sigma}
\else
\documentclass{sigma}
\fi

\numberwithin{equation}{section}

\newtheorem{Theorem}{Theorem}[section]
\newtheorem{Corollary}[Theorem]{Corollary}
\newtheorem{Lemma}[Theorem]{Lemma}

\newtheorem{Identity}[Theorem]{Identity}

\newcommand{\dd}{\mathrm{d}}
\newcommand{\e}{\mathrm{e}}
\newcommand{\dell}{\partial}
\newcommand{\C}{\mathbb{C}}
\newcommand{\R}{\mathbb{R}}
\newcommand{\Z}{\mathbb{Z}}

\def\half{{1\over 2}}

\def\sixth{{1\over6}}

\def\Eins{{\mathchoice {\rm 1\mskip-4mu l} {\rm 1\mskip-4mu l}
{\rm 1\mskip-4.5mu l} {\rm 1\mskip-5mu l}}}

\def\mn{{\mu\nu}}

\def\bef{\begin{frame}}
\def\ef{\end{frame}}
\def\benn{\begin{enumerate}\itemsep=0pt}
\def\enn{\end{enumerate}}

\def\bra#1{\langle #1 |}
\def\ket#1{| #1 \rangle}

\def\slash#1{#1\!\!\!\raise.15ex\hbox {/}}
\newcommand{\slD}{ \raise.15ex\hbox{$/$}\kern-.27em\hbox{$\!\!\!D$}}
\newcommand{\slpartial}{\raise.15ex\hbox{$/$}\kern-.57em\hbox{$\partial$}}

\newcommand{\nc}{\newcommand}

\nc{\spa}[3]{\left\langle#1 #3\right\rangle}
\nc{\spb}[3]{\left[#1 #3\right]}
\nc{\ksl}{\not{\hbox{\kern-2.3pt $k$}}}
\nc{\hf}{\textstyle{1\over2}}
\nc{\pol}{\varepsilon}
\nc{\tq}{{\tilde q}}
\nc{\esl}{\not{\hbox{\kern-2.3pt $\pol$}}}

\def\kinb{{1\over 4}\dot x^2}

\def\4piTD{{(4\pi T)}^{-{D\over 2}}}
\def\4piT4{{(4\pi T)}^{-2}}

\def\Tintm4{{\dps\int_{0}^{\infty}}\frac{{\rm d}T}{T} e^{-m^2T}
 {(4\pi T)}^{-2}}
\def\Tintm{{\dps\int_{0}^{\infty}}\frac{{\rm d}T}{T} e^{-m^2T}}

\newcommand{\slG}{{{\dot G}\!\!\!\! \raise.15ex\hbox {/}}}
\newcommand{\Gd}{{\dot G}}

\def\GBd12{{\dot G}_{B12}}

\def\dps{\displaystyle}

\begin{document}
\allowdisplaybreaks

\newcommand{\arXivNumber}{2106.12071}

\renewcommand{\thefootnote}{}

\renewcommand{\PaperNumber}{065}

\FirstPageHeading

\ShortArticleName{New Techniques for Worldline Integration}

\ArticleName{New Techniques for Worldline Integration\footnote{This paper is a~contribution to the Special Issue on Algebraic Structures in Perturbative Quantum Field Theory in honor of Dirk Kreimer for his 60th birthday. The~full collection is available at \href{https://www.emis.de/journals/SIGMA/Kreimer.html}{https://www.emis.de/journals/SIGMA/Kreimer.html}}}

\Author{James P.~EDWARDS~$^{\rm a}$, C. Moctezuma MATA~$^{\rm a}$, Uwe M\"ULLER~$^{\rm b}$ and Christian SCHUBERT~$^{\rm a}$}
\AuthorNameForHeading{J.P.~Edwards, C.M.~Mata, U.~M\"uller and C.~Schubert}

\Address{$^{\rm a)}$~Instituto de F{{\'\i}}sica y Matem\'aticas, Universidad Michoacana de San Nicol\'as de Hidalgo,\\
\hphantom{$^{\rm a)}$}~Apdo. Postal 2-82, C.P.~58040, Morelia, Michoacan, Mexico}
\EmailD{\href{mailto:james.p.edwards@umich.mx}{james.p.edwards@umich.mx},$\!$ \href{mailto:cesarfismat@gmail.com}{cesarfismat@gmail.com},$\!$
\href{mailto:christianschubert137@gmail.com}{christianschubert137@gmail.com}}

\Address{$^{\rm b)}$~Brandenburg an der Havel, Brandenburg, Germany}
\EmailD{\href{mailto:uwe.mueller.math.phys@web.de}{uwe.mueller.math.phys@web.de}}

\ArticleDates{Received March 01, 2021, in final form June 23, 2021; Published online July 03, 2021}

\Abstract{The worldline formalism provides an alternative to Feynman diagrams in the construction of amplitudes and effective actions that shares some of the superior properties of the organization of amplitudes in string theory. In~particular, it allows one to write down integral representations combining the contributions of large classes of Feynman diagrams of different topologies. However, calculating these integrals analytically without splitting them into sectors corresponding to individual diagrams poses a formidable mathematical challenge. We summarize the history and state of the art of this problem, including some natural connections to the theory of Bernoulli numbers and polynomials and multiple zeta values.}


\Classification{11B68; 33C65; 81Q30}

\renewcommand{\thefootnote}{\arabic{footnote}}
\setcounter{footnote}{0}

\section{Introduction}

When I (Christian Schubert) first met Dirk, he was still a PhD student of K.~Schilcher at Mainz University, but clearly had already set out on his quest
for new mathematical structures in~perturbative QFT, which later was to become his hallmark. At the time,
our interest in common lay in the ``$\gamma_5$-problem'', i.e., the difficulty of giving a consistent and practicable definition
for the~$\gamma_5$ matrix in the framework of dimensional renormalization. Later on it shifted to the behaviour of~perturbation theory at large orders, to the computation of renormalization group functions, and to the mathematical nature
of the objects appearing there. Although we have never actually collaborated, staying in touch with Dirk over all these years
has been a fruitful and enjoyable experience.

Our topic here will be the ``worldline formalism'', which provides
an alternative to Feynman diagrams in the construction of the perturbation series in QFT based on
first-quantized relativistic particle path integrals. Introduced by Feynman in 1950/1 for QED~\cite{feynman50,feynman51},
but then largely forgotten, it has since the nineties gained some popularity following developments in string theory~\cite{berkos1,berkos2,strassler2,strassler1}.
Although here we will be concerned only with examples from QED and scalar field theory, it should be mentioned that
worldline path integrals during the last three decades have been applied to a steadily expanding circle of problems in QFT,
providing new computational options as well as useful physical intuition (for reviews, see~\cite{126,41}).

Their non-abelian generalisation was used for a calculation
of the QCD heat-kernel coefficients to fifth order~\cite{25}
and of the two-loop effective Lagrangian for a constant ${\rm SU}(2)$ background field~\cite{sascza},
as well as for obtaining tensor decompositions of off-shell gluon amplitudes~\cite{105}.
Very recently it has also been shown how to apply it to the computation of deeply inelastic structure functions~\cite{Mueller:2019qqj,tarven}.
Yukawa and axial couplings have been included using dimensional reduction~\cite{dhogag1,dhogag2,12,16} as well as direct approaches~\cite{29}.

Extending the formalism to curved space turned out to be mathematically subtle~\cite{bastianelli,bacova1,bacova2,basvan-book,bpsv:npb446}
but eventually led to the development of new efficient methods for the calculation of amplitudes involving gravitons
\cite{125,Bastianelli:2013tsa, Bastianelli:2019xhi,Bastianelli:2012bz,Bastianelli:2002fv},
gravitational corrections to the QED effective Lagrangian~\cite{Bastianelli:2008cu},
curved-space refractive indices~\cite{Hollowood:2007ku}, and the first one-loop calculation of the photon-graviton conversion amplitude in a magnetic field~\cite{Bastianelli:2007jv,Bastianelli:2004zp}.

A well-known application of worldline path integrals is to the calculation of anomalies and index densities~\cite{alvarezgaume,alvwit,bastianelli,basvan-book,cecgir,bpsv:npb446,friwin}. In~particular, a number of special cases of the
Atiyah--Singer index theorem can be reproduced in an
elementary way by rewriting supertraces of heat kernels in terms
of supersymmetric particle path integrals
\cite{alvarezgaume,alvwit,friwin}.

Beyond standard QFT, the worldline approach has shown to be a natural tool for the construction of higher-spin field interactions
\cite{Bastianelli:2012bn, Bastianelli:2007pv}, and it lends itself to generalization to
non-commutative spaces~\cite{NCU1, NCUN, Bonezzi:2012vr, Kiem:2001dm}
as well as spaces with boundary~\cite{Bastianelli:2006hq, Bastianelli:2008vh, ConfinedScalar}.

Although most applications of the formalism are based on an analytical calculation of the path integral, efficient algorithms have also been
developed for the direct numerical computation of worldline path integrals~\cite{124, gielan}. Applications include
scalar bound states~\cite{nietjoPRL}, Casimir energies with Dirichlet boundary conditions~\cite{gilamo}, Schwinger pair-creation rates~\cite{giekli} and the
curved-space heat kernel~\cite{cormur}.

Many mathematical aspects of spacetime field theory reappear in the worldline formalism in a~one-dimensional setting,
such as worldline supersymmetry~\cite{alvarezgaume,alvwit,friwin,strassler1},
worldline instantons~\cite{afalma,63}, regularisation dependence of UV counterterms (summarized in~\cite{basvan-book}),
and Fadeev--Popov determinants induced by gauge fixing~\cite{bacozi1}.

Probably the most interesting string-related feature of the worldline formalism is that, for many cases of interest,
it allows one to derive parameter integral representations summing up whole classes of Feynman diagrams (examples will be shown in Sections~\ref{section4} and~\ref{section5} below).
However, making this fact useful for actual calculations poses a mathematical
challenge, since it leads to a highly non-standard integration problem. We will summarize the progress that has been achieved
along these lines, including the solution of the problem at the polynomial level, and explain a recent algorithm
that makes it possible to analytically compute the worldline integral representing the standard one-loop scalar $N$-point diagram
together with all the ``crossed'' diagrams. A central role in worldline integration is played by the Bernoulli polynomials, which
over the years has also led to some number-theoretical spin-off, some of which will be discussed,~too.

\section{Feynman's worldline formalism}\label{section2}

We start with a short review of Feynman's worldline formalism. Consider the
Green's function for the covariantized Klein--Gordon operator $(\partial + {\rm i}\e A)^2 + m^2$,
\begin{gather*}
D^{xx'}[A] \equiv
\bigg\langle x' \bigg| \frac{1}{-(\partial + {\rm i}\e A)^2 + m^2} \bigg| x \bigg\rangle .
\end{gather*}
We work with euclidean conventions, and as usual set $\hbar = c=1$.
We exponentiate the denominator using a Schwinger proper-time parameter $T$:
\begin{gather*}
D^{xx'} [A] =\bigg\langle x'\bigg| \int_0^{\infty}{\rm d}T \exp \bigl[
{-} T \big( {-}(\partial + {\rm i}\e A)^2+m^2\big)\bigr] \bigg| x \bigg\rangle .
\end{gather*}
By a standard discretization procedure, the $x$-space matrix element can be transformed into a~path integral,
\begin{gather*}
D^{xx'} [A] =\int_0^{\infty}{\rm d}T\e^{-m^2T}
\int_{x(0)=x}^{x(T)=x'}{\cal D}x(\tau)\e^{-\int_0^T {\rm d}\tau \left(\kinb +{\rm i}\e\dot x\cdot A(x(\tau))\right)} .
\end{gather*}
Choosing the background field as a sum of $N$ plane waves,
\begin{gather}
A^{\mu}(x(\tau)) = \sum_{i=1}^N \varepsilon_i^{\mu}\e^{{\rm i}k_i\cdot x(\tau)}
\label{pw}
\end{gather}
and Fourier transforming the endpoints we get the ``photon-dressed propagator'', as shown in~Fi\-gure~\ref{fig-propexpand},
where for our present purposes it must be emphasized that
summation over the $N!$ permutations of the photons is understood.

\begin{figure}[htbp]\centering
 \includegraphics[scale=0.82]{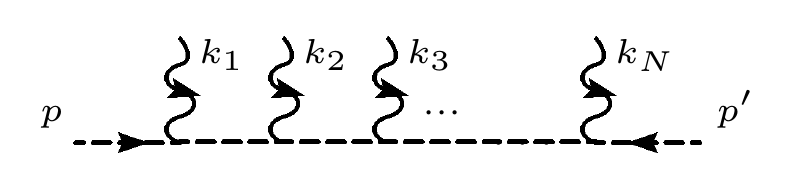}
\caption{Scalar propagator dressed with $N$-photons.}\label{fig-propexpand}
\end{figure}

For the one-loop effective action, the same procedure gives
\begin{align}
\Gamma [A]&=- \operatorname{Tr}{\rm ln} \bigl\lbrack {-}(\partial + {\rm i}\e A)^2+m^2\bigr\rbrack
= \int_0^{\infty} \frac{{\rm d}T}{T} \operatorname{Tr} \exp \bigl\lbrack {-} T \big( {-}(\partial + {\rm i}\e A)^2+m^2\big)\bigr\rbrack \nonumber
\\
&=\int_0^{\infty}\frac{{\rm d}T}{T}\e^{-m^2T}
\int_{x(0)=x(T)}{\cal D}x(\tau)\e^{-\int_0^T {\rm d}\tau \left(\kinb +{\rm i}\e\dot x\cdot A(x(\tau))\right)} .
\label{EAscal}
\end{align}
Upon specialization to the plane-wave background \eqref{pw} and an expansion to $N$th order in the
coupling, one obtains the one-loop $N$, photon amplitudes in the form
\begin{gather}
\Gamma[\lbrace k_i,\varepsilon_i\rbrace]=(-{\rm i}\e)^{N}\!\!
\int\! \frac{{\rm d}T}{T} \e^{-m^2T}\!\!\int \!{\cal D}x
V_{\rm scal}^\gamma[k_1,\!\varepsilon_1]\cdots
V_{\rm scal}^\gamma[k_N,\!\varepsilon_N] \e^{-\int_0^T{\rm d}\tau {\dot x^2\over 4}}.
\label{Nvertop}
\end{gather}
Here $V_{\rm scal}^\gamma$ denotes the same photon
vertex operator as is used in string perturbation
theory,
\begin{gather*}
V_{\rm scal}^\gamma[k,\varepsilon]
=
\int_0^T{\rm d}\tau
\varepsilon\cdot \dot x(\tau)
 {\rm e}^{{\rm i}kx(\tau)} .
\end{gather*}
From these two building blocks, arbitrary scalar QED amplitudes can be constructed by sewing pairs of photons. See, e.g.,
Figure~\ref{fig-QEDSmatrix}.

\begin{figure}[h!]\centering
 \includegraphics[scale=.50]{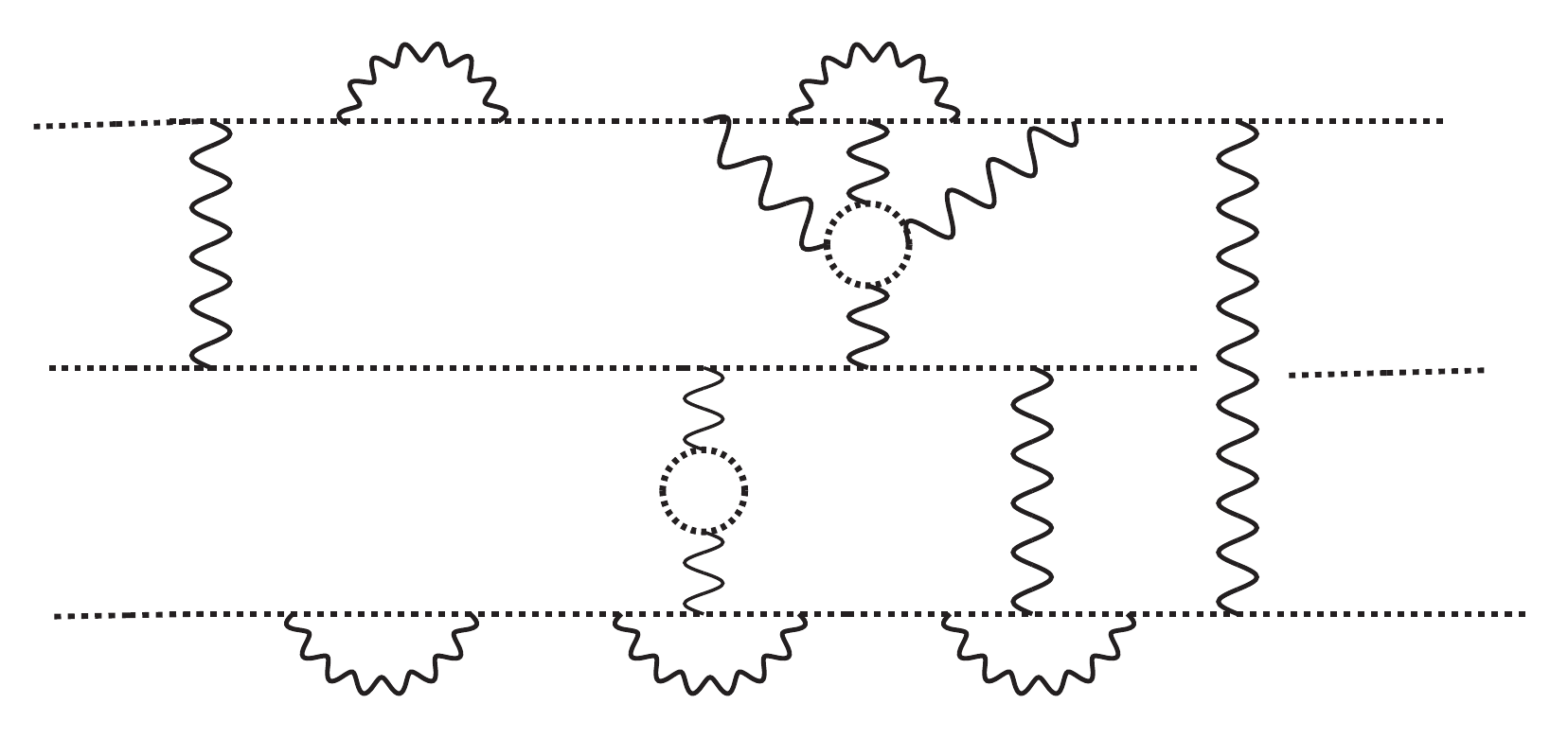}
\caption{A typical multiloop diagram in scalar QED.}\label{fig-QEDSmatrix}
\end{figure}

An advantage of this formalism is that scalar QED and spinor QED become more closely related than usual, at least in
perturbation theory.
For example, to get the effective action for spinor QED from \eqref{EAscal} requires only a change of the global
normalization by a factor of $-\half$, and the insertion of the following {\it spin factor} ${\rm Spin}[x,A]$ under the path integral,
\begin{gather}
{\rm Spin}[x,A] = {\rm tr}_{\Gamma} {\cal P}
\exp\biggl[\frac{\rm i}{4}e [\gamma^{\mu},\gamma^{\nu}]
\int_0^T{\rm d}\tau F_{\mu\nu}(x(\tau))\biggr].
\label{spinfactor}
\end{gather}
Here $\cal P$ indicates that the exponential function in general will not be an
ordinary but path-ordered one, since the exponents at different proper times will not commute
in general.

\section{Modern approach to the worldline formalism}\label{section3}

Nowadays for analytical purposes one usually does not use Feynman's spin factor in the form \eqref{spinfactor},
but replaces it by the following Grassmann path integral~\cite{fradkin66}:
\begin{gather*}
{\rm Spin}[x,A] \to\int {\cal D}\psi(\tau)\exp
\biggl\lbrack-\int_0^T{\rm d}\tau\biggl(
\half \psi\cdot \dot\psi -{\rm i}\e \psi^{\mu}F_{\mn}\psi^{\nu}\biggr)\biggr\rbrack .
\end{gather*}
Here $\psi^\mu (\tau)$ is a Lorentz-vector valued anticommuting function of the proper-time,
\begin{gather*}
\psi(\tau_1)\psi(\tau_2) = - \psi(\tau_2)\psi(\tau_1)
\end{gather*}
and it has to be functionally integrated using antiperiodic boundary conditions, $\psi(T) = - \psi(0)$.
Apart from removing the path-ordering, which will be essential for the applications to be discussed below,
it has the further advantage of exhibiting a supersymmetry between the orbital and spin degrees of freedom
of the electron (generalizing the well-known supersymmetry of the Pauli equation in quantum mechanics).
A similar approach can be followed in the non-abelian case to eliminate the path ordering due to the presence of colour factors
\cite{Ahmadiniaz:2015xoa, Balachandran:1976ya,Barducci:1976xq,Bastianelli:2013pta}.

However, it was only in the nineties, after the development of string theory, which demonstrated the
mathematical beauty and computational usefulness of first-quantized path integrals,
that Feynman's worldline path-integral formalism was finally taken seriously as a
competitor for Feynman diagrams. In~this ``string-inspired'' approach to the worldline formalism
\cite{berkos1,berkos2,strassler2,strassler1}, a~central role is played by gaussian path integration
using ``worldline Green's functions''. In~QED, the only worldline correlators required are
$G(\tau_1,\tau_2)$, $G_F(\tau_1,\tau_2)$,
obeying
\begin{gather}
\langle x^{\mu}(\tau_1)x^{\nu}(\tau_2) \rangle=-G(\tau_1,\tau_2) \delta^{\mu\nu}, \qquad
G(\tau_1,\tau_2) = \vert \tau_1 -\tau_2\vert - \frac{1}{T}\bigl(\tau_1 -\tau_2\bigr)^2,
\nonumber
\\
\langle \psi^{\mu}(\tau_1)\psi^{\nu}(\tau_2)\rangle=\half G_F(\tau_1,\tau_2) \delta^{\mu\nu} , \qquad
G_F(\tau_1,\tau_2) = {\rm sgn}(\tau_1 - \tau_2)
\label{defG}
\end{gather}
(note that in the worldline formalism ${\rm sgn}(0)=0$ in general).
For example, to calculate the path-integral \eqref{Nvertop}, first we have to make the kinetic
operator invertible, which requires eliminating the constant trajectories. This is best done by separating off the
average position of the loop
\begin{gather*}
x_0^\mu\equiv \frac{1}{T}\int_0^T{\rm d}\tau x^\mu(\tau),
\end{gather*}
whose integral will
just produce the global energy-momentum conservation factor $(2\pi)^D\delta(\sum k_i)$.
It then requires only a formal exponentiation $\varepsilon\cdot \dot x(\tau)
 {\rm e}^{{\rm i}kx(\tau)} = \e^{{\rm i}kx+\varepsilon\cdot \dot x(\tau)}\mid_{{\rm lin}(\varepsilon)}$
to arrive at a~gaussian path integral, which can be performed by a simple completion of the square.
This leads to the prototype of a ``Bern--Kosower master formula'',
\begin{gather}
\Gamma[\lbrace k_i,\varepsilon_i\rbrace]
={(-{\rm i}\e)}^N{\dps\int_{0}^{\infty}}\frac{{\rm d}T}{T}
{(4\pi T)}^{-\frac{D}{2}}e^{-m^2T}\prod_{i=1}^N \int_0^T{\rm d}\tau_i
\nonumber
\\ \hphantom{\Gamma[\lbrace k_i,\varepsilon_i\rbrace]=}
{}\times
\exp\bigg\lbrace\sum_{i,j=1}^N
\bigg\lbrack \half G_{ij} k_i\cdot k_j
+{\rm i}\dot G_{ij}k_i\cdot\varepsilon_j
+\half\ddot G_{ij}\varepsilon_i\cdot\varepsilon_j
\bigg\rbrack\bigg\rbrace
\bigg|_{{\rm lin}(\varepsilon_1,\ldots,\varepsilon_N)},
\label{bk}
\end{gather}
which is somewhat formal since, after deciding on the number of photons, one still has to expand
the exponential factor and to keep only the terms that contain each of the $N$ polarization vectors once. Besides the Green's function \eqref{defG}, now also its first two derivatives appear,
\begin{gather}
\dot G(\tau_1,\tau_2)= {\rm sgn}(\tau_1 - \tau_2)
- 2 \frac{{(\tau_1 - \tau_2)}}{T} ,\nonumber
\\
\ddot G(\tau_1,\tau_2)= 2 {\delta}(\tau_1 - \tau_2)- \frac{2}{T},
\label{Gddot}
\end{gather}
where it is understood that a ``dot'' always refers to a derivative with respect to the first variable.
The factor ${(4\pi T)}^{-{D\over 2}}$ comes from the free path integral determinant.

The master formula \eqref{bk} has many attractive features, some obvious, some less so:
\benn
\item
It provides a highly compact generating function for the $N$-photon amplitudes in scalar QED, valid off-shell and on-shell.

\item
It represents the sum of the corresponding Feynman diagrams including all non-equivalent orderings of the photons
along the loop.

\item
Bern and Kosower in their seminal work~\cite{berkos1,berkos2} found a set of rules that allows one to obtain from this master formula,
by a pure pattern matching procedure, also the corresponding amplitudes with a spinor loop, as well as the $N$-gluon amplitudes
with a scalar, spinor or gluon loop.

\item
In this formalism, the worldline Lagrangian contains only a linear coupling to the background field, corresponding to
a cubic vertex in field theory. The quartic seagull vertex arises only at the path-integration stage, and is represented
by the $\delta(\tau_i - \tau_j)$ contained in~$\ddot G_{ij}$, equation~\eqref{Gddot}.

\item
All the $\ddot G_{ij}$ can be removed by a systematic integration-by-parts procedure, which homo\-ge\-nizes the integrand
and at the same time leads to the appearance of photon field strength tensors
$f_i^{\mu\nu} \equiv k_i^{\mu}\varepsilon_i^{\nu} - \varepsilon_i^{\mu}k_i^{\nu}$,
as was noted by Strassler~\cite{strassler2}.
\enn

\section{The four-photon amplitudes}\label{section4}

Let us have a closer look at the four-photon case, which is important not only for QED but also serves as
the prototypical example for all amplitudes involving four gauge bosons. In~terms of~Feynman diagrams,
in spinor QED it is given by the familiar six diagrams shown in Figure~\ref{fig-photonphoton}, while in
scalar QED there are a few more diagrams involving the seagull vertex.

\begin{figure}[htbp]\centering
\includegraphics[scale=0.9]{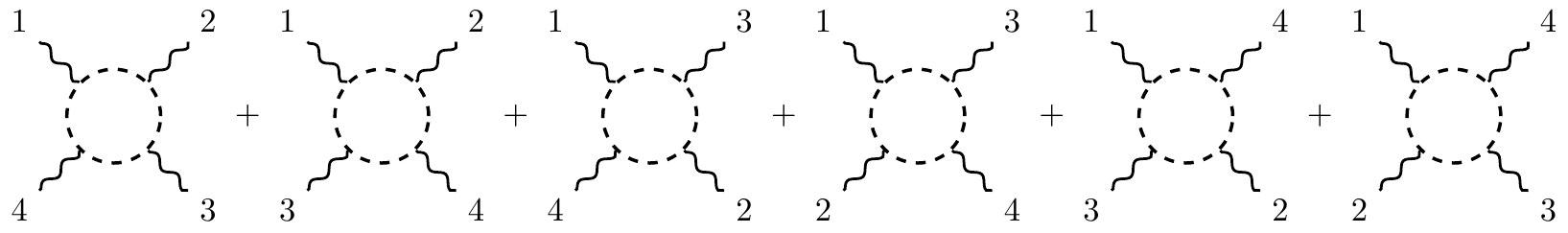}
\caption{Feynman diagrams for photon-photon scattering.}\label{fig-photonphoton}
\end{figure}

At the four-point level, the requirement of removing all the
$\ddot G_{ij}$ by no means exhausts the freedom of integration by parts, even if requesting that the full permutation symmetry
between the photons be preserved, and it is only very recently~\cite{136}
that an algorithm was found that is optimized in the sense that, apart from achieving the removal of all $\ddot G_{ij}$,
it also reduces the number of independent tensor structures arising to its minimum possible, which is five. It leads to the
following decomposition:
\begin{gather*}
\hat \Gamma_4 = \hat \Gamma_4^{(1)} + \hat \Gamma_4^{(2)} + \hat \Gamma_4^{(3)} + \hat \Gamma_4^{(4)} + \hat \Gamma_4^{(5)},
\\
\hat \Gamma_4^{(1)} = \hat \Gamma^{(1)}_{(1234)}T^{(1)}_{(1234)}
 + \hat \Gamma^{(1)}_{(1243)}T^{(1)}_{(1243)}
 + \hat \Gamma^{(1)}_{(1324)}T^{(1)}_{(1324)} ,
\\[.5ex]
 \hat \Gamma_4^{(2)} = \hat \Gamma^{(2)}_{(12)(34)}T^{(2)}_{(12)(34)}
 + \hat \Gamma^{(2)}_{(13)(24)}T^{(2)}_{(13)(24)}
 + \hat \Gamma^{(2)}_{(14)(23)}T^{(2)}_{(14)(23)},
 \\[.5ex]
 \hat \Gamma_4^{(3)}= \sum_{i=1,2,3} \hat \Gamma^{(3)}_{(123)i}T^{(3)r_4}_{(123)i}
 + \sum_{i=2,3,4} \hat \Gamma^{(3)}_{(234)i}T^{(3)r_1}_{(234)i}
 + \sum_{i=3,4,1} \hat \Gamma^{(3)}_{(341)i}T^{(3)r_2}_{(341)i}
 + \sum_{i=4,1,2} \hat \Gamma^{(3)}_{(412)i}T^{(3)r_3}_{(412)i} ,
 \\
\hat \Gamma_4^{(4)} =
\sum_{i<j} \hat \Gamma^{(4)}_{(ij)ii}T^{(4)}_{(ij)ii} +
\sum_{i<j} \hat \Gamma^{(4)}_{(ij)jj}T^{(4)}_{(ij)jj},
\\
\hat \Gamma_4^{(5)} =\sum_{i<j} \hat \Gamma^{(5)}_{(ij)ij}T^{(5)}_{(ij)ij} +
\sum_{i<j} \hat \Gamma^{(5)}_{(ij)ji}T^{(5)}_{(ij)ji}
\end{gather*}
with a set of five tensors $T^{(i)}$ that, remarkably, up to normalization agrees
with the one found in 1971 by Costantini, De Tollis and Pistoni~\cite{cotopi} using the QED Ward identity,\vspace{-1ex}
\begin{gather*}
T^{(1)}_{(1234)} \equiv Z_4(1234),
\\
T^{(2)}_{(12)(34)} \equiv Z_2(12)Z_2(34),
\\
T^{(3)r_4}_{(123)i} \equiv Z_3(123) \frac{r_4\cdot f_4\cdot k_i }{r_4\cdot k_4}, \qquad i=1,2,3 ,
\\
T^{(4)}_{(12)ii} \equiv Z_2(12) \frac{k_i \cdot f_3 \cdot f_{4} \cdot k_i}{k_3\cdot k_4}, \qquad i=1,2 ,
\\
T^{(5)}_{(12)ij} \equiv Z_2(12) \frac{k_i \cdot f_3 \cdot f_{4} \cdot k_j}{k_3\cdot k_4},\qquad (i,j)=(1,2),(2,1),
\end{gather*}
where we abbreviated\vspace{-1ex}
\begin{gather*}
Z_2(ij)\equiv
\half {\rm tr}\bigl(f_if_j\bigr) = \varepsilon_i\cdot k_j\varepsilon_j\cdot k_i - \varepsilon_i\cdot\varepsilon_jk_i\cdot k_j,
\\
Z_n(i_1i_2\cdots i_n)\equiv {\rm tr}\bigg(\prod_{j=1}^nf_{i_j}\bigg),
\qquad n\geq 3.
\end{gather*}
Note that the tensor $T^{(3)r_4}_{(123)i}$ still depends on a ``reference momentum''~$r_4$, which is arbitrary except for the
condition $r_4\cdot k_4 \ne 0$.
The coefficient functions are given by~\cite{136}
\begin{gather*}
\hat\Gamma^{(k)}_{\cdots}=	\int_0^\infty \frac{{\rm d}T}{T} T^{4-\frac{D}{2}}\e^{-m^2T}
\int_0^1\prod_{i=1}^4{\rm d}u_i \hat \gamma^{(k)}_{\ldots}\big(\Gd_{ij}\big)	 \e^{T\sum_{i<j=1}^4G_{ij}k_i\cdot k_j},
\end{gather*}
where, for spinor QED
\begin{gather*}
\hat \gamma^{(1)}_{(1234)} = \Gd_{12}\Gd_{23}\Gd_{34}\Gd_{41} - G_{F12}G_{F23}G_{F34}G_{F41},
\\
\hat \gamma^{(2)}_{(12)(34)} = \bigl(\Gd_{12}\Gd_{21} - G_{F12}G_{F21}\bigr) \bigl(\Gd_{34}\Gd_{43} - G_{F34}G_{F43}\bigr),
\\
\hat \gamma^{(3)}_{(123)1} = \bigl(\Gd_{12}\Gd_{23}\Gd_{31} - G_{F12}G_{F23}G_{F31}\bigr) \Gd_{41},
\\
\hat \gamma^{(4)}_{(12)11} = \bigl(\Gd_{12}\Gd_{21} - G_{F12}G_{F21}\bigr) \Gd_{13}\Gd_{41},
 \\
\hat \gamma^{(5)}_{(12)12} = \bigl(\Gd_{12}\Gd_{21} - G_{F12}G_{F21}\bigr) \Gd_{13}\Gd_{42}
\end{gather*}
(plus permutations thereof),
and the coefficient functions for scalar QED are obtained from these simply by deleting all the $G_{Fij}$. In~a forthcoming series of papers~\cite{4photon} this representation is used
for a first calculation of the scalar and spinor QED four-photon amplitudes completely off-shell.
The results should become useful for higher-order calculations in QED with four-photon sub-diagrams,
as well as for photonic processes in external fields.

\section{On to multiloop}\label{section5}

For example, from the four-photon amplitude we can, by sewing, construct the two-loop photon propagator, Figure~\ref{fig-2loopvpdiag}
\begin{figure}[htbp]\centering
\includegraphics[scale=0.6]{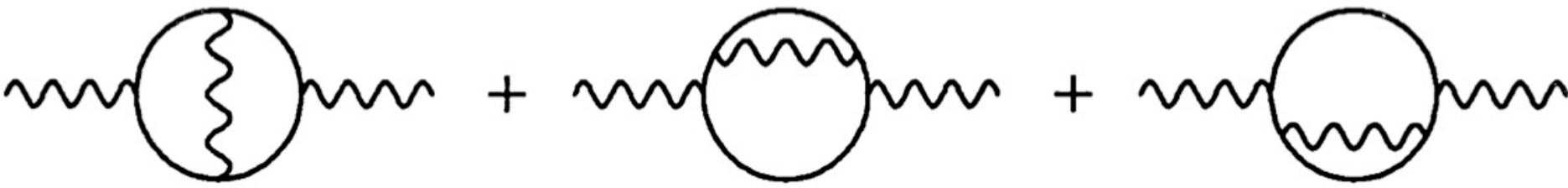}
\caption{The two-loop photon propagator.}\label{fig-2loopvpdiag}
\end{figure}

\noindent
and from the one-loop six-photon amplitude we get the three-loop quenched propagator, Fi\-gure~\ref{fig-3loopbetadiag},

\begin{figure}[htbp]\centering
\includegraphics[scale=0.6]{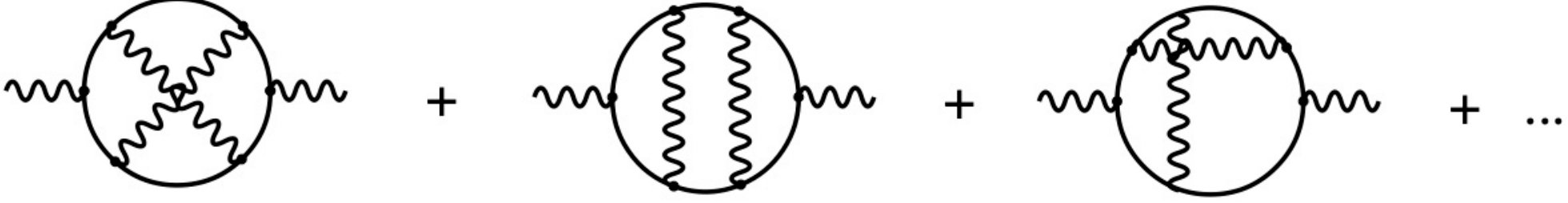}
\caption{The three-loop quenched photon propagator.}\label{fig-3loopbetadiag}
\end{figure}
\noindent
etc. The interesting feature of the worldline formalism is that the sewing results in parameter integrals that
represent not some particular Feynman diagram, but the whole set of Feynman diagrams shown in Figures~\ref{fig-2loopvpdiag} and~\ref{fig-3loopbetadiag}.
And it is precisely this type of sums of diagrams that have become notorious in the QED community for the particularly
extensive cancellations between diagrams that were discovered as a byproduct of the enormous effort
that was invested from the sixties onward in obtaining information on the high-order behaviour of the QED $\beta$ function.
Both physically and mathematically, the most interesting contributions to this function come from the
quenched (one fermion-loop) diagrams, and this ``quenched QED $\beta$ function''
was a big topic at the Multiloop Workshop that took place at the Aspen Center of Physics in 1995,
where besides Dirk also D. Broadhurst, K. Chetyrkin, A. Davydychev and other experts in multiloop integration were present.
The coefficients known at the time (up to four loops in spinor QED,
up to three loops in scalar QED) were all rational~\cite{brdekr}:
\begin{gather*}
\beta_3^{\rm spin} = -2
\end{gather*}
after a cancellation of $\zeta(3)$ between diagrams,
\begin{gather*}
\beta_4^{\rm spin} = - 46
\end{gather*}
after a cancellation of $\zeta(3)$s and $\zeta(5)$s. And in scalar QED, too, the $\zeta(3)$ contributions cancelled at three loops,
\begin{gather*}
\beta_3^{\rm scal}= \frac{29}{2}.
\end{gather*}
Thus at the time everybody believed that the quenched coefficients would stay rational to all loop orders
and eventually might be computable in closed form, but then many years later at five loops
$\zeta(3)$ refused to cancel out in $\beta_5^{\rm spin}$~\cite{bckr}.
The Lord sometimes {\it can} be malicious.
Still, although incomplete the cancellations are remarkable and ought to find an explanation!

\section{The fundamental problem of worldline integration}\label{section6}

Returning to the one-loop level, using that $G_{Fij}$s always appear in closed chains and can be eliminated by
\begin{gather*}
G_{Fij}G_{Fjk}G_{Fki} = - \big(\dot G_{ij} + \dot G_{jk} + \dot G_{ki}\big)
\end{gather*}
 the most general integral that one will ever have to compute in the worldline approach to QED
is of the form
\begin{gather*}
\int_0^1{\rm d}u_1{\rm d}u_2 \cdots {\rm d}u_N {\rm Pol} \big(\dot G_{ij}\big) \e^{\sum_{i<j=1}^N \lambda_{ij}^2 G_{ij} }
\end{gather*}
with arbitrary $N$ and polynomial ${\rm Pol}\big(\dot G_{ij}\big)$, where
\begin{gather*}
G_{ij} = |u_i-u_j| - (u_i-u_j)^2, \qquad
\dot G_{ij} = {\rm sgn}(u_i-u_j) - 2(u_i-u_j) .
\end{gather*}
For this to become eventually useful in addressing the type of multiloop cancellations that we have just reviewed,
one should devise a way of performing such integrals {\it without decomposing the integrand into ordered sectors}.

At the polynomial level, this turned out to be a quite tractable problem. By Fourier analysis, it is easy to show that
integrating a polynomial expression in $\dot G_{ij}$ ``full circle'' in any particular variable $u_i$ yields again a
polynomial expression in the remaining $\dot G_{ij}$. To give just one particularly neat example,
\begin{gather*}
\int_0^1 {\rm d}u\,\dot G(u,u_1)\dot G(u,u_2)\dot G(u,u_3)
=-\sixth\big(\Gd_{12}-\Gd_{23}\big)\big(\Gd_{23}-\Gd_{31}\big)\big(\Gd_{31}-\Gd_{12}\big).
\end{gather*}
However, the result can generally be written in many different equivalent ways due to identities such as
\begin{gather*}
\big(\dot G_{ij}+\dot G_{jk} +\dot G_{ki}\big)^2 = 1 .
\end{gather*}
Eventually, using the master integral
\begin{gather*}
\int_0^1 {\rm d}u\,\e^{\sum_{i=1}^n c_i \dot G(u,u_i)}
=\frac{\sum_{i=1}^n \sinh (c_i)\e^{\sum_{j=1}^n c_j \dot G_{ij}}}{\sum_{i=1}^n c_i}
\end{gather*}
with constants $c_1,\ldots,c_n$
the following formula was found which allows one to integrate an arbitrary monomial in $\dot G_{ij}$
in one of the variables, and gives the result as a polynomial in the remaining $\dot G_{ij}$:
\begin{gather*}
\int_0^1 {\rm d}u\,\dot G(u,u_1)^{k_1}\dot G(u,u_2)^{k_2}\cdots\dot G(u,u_n)^{k_n}
\\ \qquad
{}=
\frac{1}{2n}\sum_{i=1}^n\prod_{j\ne i}\sum_{l_j=0}^{k_j}{k_j\choose l_j}
\dot G_{ij}^{k_j-l_j}\sum_{l_i=0}^{k_i}{k_i\choose l_i}
\frac{(-1)^{\sum_{j=1}^n l_j}}{(1+\sum_{j=1}^n l_j)n^{\sum_{j=1}^n l_j}}
\\ \qquad\phantom{=}
{}\times\Biggr\lbrace\!\bigg(\sum_{j\ne i}\dot G_{ij} +1 \bigg)^{1+\sum_{j=1}^n l_j}
- (-1)^{k_i-l_i}\bigg(\sum_{j\ne i}\dot G_{ij} -1\bigg)^{1+\sum_{j=1}^n l_j}\Biggr\rbrace.
\end{gather*}
Thus by recursion it provides a complete solution to the ``circular integration'' problem at the polynomial level.

\section{Bernoulli numbers and polynomials}\label{section7}

Next, let us emphasize that worldline integration naturally relates to the theory of Bernoulli numbers and polynomials.
This is because the coordinate path integral was performed in the Hilbert space $ H_P'$ of periodic functions orthogonal to the constant functions (because of the elimination of the zero mode). In~this space the ordinary $n$th derivative $\partial_P^n$ is invertible, and the integral
kernel of the inverse is given essentially by the $ n$th Bernoulli polynomial $ B_n(x)$:
\begin{gather}
\langle u\mid {\partial}^{-n}_P \mid u'\rangle
= -\frac{1}{n!}B_n(\vert u-u'\vert){\rm sgn}^n(u-u'),\qquad n\geq 1,
\label{deln}
\\
\bra{u} \partial^{0}\ket{u'} = \delta(u-u') -1.
\label{invder}
\end{gather}
This fact is related to the well-known Fourier series
\begin{gather*}
B_n(x) = - \frac{n!}{(2\pi {\rm i})^n} \sum_{k\ne 0} \frac{\e^{2\pi {\rm i} k x}}{k^n}
\end{gather*}
although in the mathematical literature it is usually assumed that $0<x<1$, which eliminates the $\rm sgn$ and $\delta$
functions in \eqref{deln}, \eqref{invder}, that is, just the fun stuff that makes these inverse derivatives
into a well-defined algebra of integral operators in the Hilbert space $ H_P'$.

Note that the left-hand side of \eqref{deln} can also be written in terms of the $\dot G_{ij}$ as a ``chain integral'',
\begin{gather*}
\big\langle u_1\mid {\partial}^{-n}_P \mid u_{n+1}\big\rangle = \frac{1}{2^n}
\int_0^1 {\rm d}u_2\cdots {\rm d}u_n\dot G_{12}\dot G_{23}\cdots\dot G_{n(n+1)}
\end{gather*}
since the integral just constructs a multiple inverse derivative by iteration of a single one.

Thus the ubiquitous appearance of the Bernoulli numbers $B_n$ $(=B_n(0))$ in perturba\-tive~QFT, which from the diagrammatic point of view
often remains somewhat mysterious, in the worldline formalism finds a natural explanation.

As a nice example for a non-trivial occurrence of the Bernoulli numbers in QED amplitudes,
let us show here the ``all +'' amplitudes in scalar and spinor QED,
in the low-energy (LE) limit, at one and two loops~\cite{51,56}
\begin{gather*}
\Gamma_{\rm spin}^{(1,2)(LE)}({\rm all} +)=-\frac{\alpha \pi m^4}{8\pi^2}
\bigg(\frac{2e}{m^2}\bigg)^N
 c_{\rm spin}^{(1,2)}\bigg(\frac{N}{2}\bigg)\chi_N^+,
\\
\Gamma_{\rm scal}^{(1,2)(LE)}({\rm all} +)={\alpha\pi m^4\over 16\pi^2}
\bigg(\frac{2e}{m^2}\bigg)^N c_{\rm scal}^{(1,2)}\bigg(\frac{N}{2}\bigg)\chi_N^+.
\end{gather*}
Here $\chi_N^+$ is a spinorial invariant that absorbs all the momenta and polarization vectors, and can
be constructed explicitly in the spinor-helicity formalism, see~\cite{56}.
The one-loop coefficients are equal,
\begin{gather*}
c^{(1)}_{\textrm{spin}}(n) = \frac{(-1)^{n+1}B_{2n}}{2n(2n-2)} = c^{(1)}_{\textrm{scal}}\left(n \right)
\end{gather*}
however this degeneracy is lifted at the two-loop level,
\begin{gather*}
c^{(2)}_{\rm spin}(n) =\frac{1}{(2\pi)^2}\biggl\lbrace
\frac{2n-3}{2n-2} {B}_{2n-2}
+3\sum_{k=1}^{n-1}\frac{{B}_{2k}}{2k}
\frac{{B}_{2n-2k}}{2n-2k}\biggr\rbrace,
\\
c^{(2)}_{\rm scal}(n)=\frac{1}{(2\pi)^2}\biggl\lbrace
\frac{2n-3}{2n-2} {B}_{2n-2}+\frac{3}{2}\sum_{k=1}^{n-1}
\frac{{B}_{2k}}{2k}\frac{{B}_{2n-2k}}{2n-2k}
\biggr\rbrace.
\end{gather*}
(The two-loop result is actually the quenched contribution only, there is also a non-quenched one~\cite{118,giekar}.)

\section{The Miki and Faber--Pandharipande--Zagier identities}\label{section8}

The above two-loop result actually led to some mathematical spin-off which merits a digression.
It was not obtained by a direct calculation of the $N$-photon amplitudes, but by a calculation of the scalar and spinor
QED effective Lagrangian at two loops in a (Euclidean) self-dual background field by G.V.~Dunne and one of the authors in 2002~\cite{51,52}.
Let us now focus on the scalar QED case. Doing this calculation in two different ways, both using the worldline formalism,
we obtained there two different formulas in terms of Bernoulli numbers for the same weak-field expansion coefficients,
\begin{gather}
c^{(2)}_{\rm scal}(n)=\frac{1}{(2\pi)^2}\biggl\lbrace
\frac{2n-3}{2n-2} B_{2n-2}+\frac{3}{2}\sum_{k=1}^{n-1}
\frac{B_{2k}}{2k}\frac{B_{2n-2k}}{2n-2k}\biggr\rbrace
\label{cn2app}
\end{gather}
vs.
\begin{gather}
c^{(2)}_{\rm scal}(n)=\frac{1}{(2\pi)^2}\biggl\lbrace
\frac{2n-3}{2n-2} B_{2n-2}
+3\bigg[\psi(2n+1)-\frac{2}{2n-1}+\gamma-1\bigg]\frac{B_{2n}}{2n}\nonumber
\\ \hphantom{c^{(2)}_{\rm scal}(n)=\frac{1}{(2\pi)^2}\biggl\lbrace}
{}+3\sum_{k=1}^{n-1}{{2n-2}\choose{2k-2}}\frac{B_{2k}}{2k}\frac{B_{2n-2k}}{2n-2k}\biggr\rbrace.
\label{cn2alt}
\end{gather}
Here
$\psi(x)=\Gamma'(x)/\Gamma(x)$, and $\gamma$ is Euler's constant.

Although there is a rather enormous mathematical literature on identities involving Bernoulli numbers,
we could not find anything that would have allowed us to show the equivalence of~these two expressions.
Dirk suggested to consult Richard Stanley from MIT, a leading expert in~combinatorics, and indeed he
showed us how to demonstrate the equivalence by combining the probably best-known of Bernoulli number identities, Euler's identity
\begin{gather*}
\sum_{k=1}^{n-1}{2n\choose 2k}B_{2k}B_{2n-2k}=-(2n+1)B_{2n}
\end{gather*}
with a much deeper identity due to Miki (1977):
for integer $n\geq 2$,
\begin{gather*}
\sum_{k=1}^{n-1}\frac{B_{2k}B_{2n-2k}}{2k(2n-2k)}
=\sum_{k=1}^{n-1}\frac{B_{2k}B_{2n-2k}}{2k(2n-2k)}{2n\choose 2k} +\frac{B_{2n}}{n}H_{2n},
\end{gather*}
where $H_i$ denotes the $i$th harmonic number.

At about the same time, we learned about a similar identity that was found heuristically by Faber and Pandharipande in 1998
in a string theory calculation, and then proven by Zagier~\cite{fapaza}: for integer $n\geq 2$,
\begin{gather*}
\sum_{k=1}^{n-1}\frac{\bar{B}_{2k}\bar{B}_{2n-2k}}{2k(2n-2k)}
=\frac{1}{n}\sum_{k=1}^{n}\frac{B_{2k}\bar{B}_{2n-2k}}{2k}{2n\choose 2k}
+\frac{\bar{B}_{2n}}{n}H_{2n-1},
\end{gather*}
where we have introduced the further abbreviation
\begin{gather*}
\bar{B}_n\equiv \bigg(\frac{1-2^{n-1}}{2^{n-1}}\bigg)B_n.
\end{gather*}

In~\cite{59}, we reanalyzed our worldline derivation of \eqref{cn2app} and \eqref{cn2alt},
and used it as a guiding principle to
\benn
\item
Give a unified proof of the Miki and Faber--Pandharipande--Zagier identities.

\item
To generalize them at the quadratic level.

\item
Generalize them to the cubic level.

\item
Outline the construction of even higher-order identities.
\enn

\section{A toy QFT for multiple zeta-value identities}\label{section9}

As a further mathematical digression, let me mention that the above algebra of inverse derivatives
can be mathematically enriched by changing the Hilbert space from $H_P'$ to one ``chiral'' half,
that is, the Hilbert space $H_P^+$ generated by the basis $\e^{2\pi {\rm i} k x}$ with positive integers $k$.
The inverse derivative $\langle u\mid {\partial}^{-n}_P \mid u'\rangle$ then acquires also an imaginary
part, which up to normalization is the $n$th Clausen function ${\rm Cl}_n(2\pi (u-u'))$.

In~\cite{34} two of the authors along these lines constructed a worldline QFT geared towards the derivation of identities
between multiple zeta values, defined by
\begin{gather*}
\zeta(k_1,\ldots,k_m) = \sum_{n_1>n_2>\cdots >n_m >0} \frac{1}{n_1^{k_1}\cdots n_m^{k_m}} .
\end{gather*}
In this model, arbitrary such values can be represented as ``seashell'' Feynman diagrams, and
large classes of identities between them can be derived through a systematic integration-by-parts
procedure applied to the corresponding Feynman integrals.

\section[General worldline integral in phi3 theory]{General worldline integral in $\boldsymbol{\phi^3}$ theory}\label{section10}

Finally, let us return to the problem of circular integration, and our ongoing work, where we are trying to~solve it for the simplest non-polynomial case, the off-shell one-loop $N$-point amplitude for
scalar $\phi^3$-theory in $D$ dimensions. For this case the worldline master formula reads~\cite{polbook},
\begin{gather*}
I_N(p_1,\ldots,p_N)=
\half (4 \pi)^{-D/2} (2\pi)^D\delta\bigg(\sum_{i=1}^N p_i\bigg)\hat I_N(p_1,\ldots,p_N),
\\
\hat I_N(p_1,\ldots,p_N)
=\int_0^\infty \frac{{\rm d}T}{T}T^{N-D/2} e^{-m^2 T} \int_{12\ldots N}\e^{T\sum_{i<j=1}^N G_{ij} p_{ij}},
\end{gather*}
where we abbreviated
$\int_{12\ldots N}\equiv \int_0^1 {\rm d}u_1 \cdots \int_0^1 {\rm d}u_N$
and~$p_{ij} \equiv p_i\cdot p_j$.

The ordered integrals lead off-shell to hypergeometric functions such as
$_2F_1$ (2-point), $F_1$ (3-point), the Lauricella--Saran function (4-point), see, e.g.,~\cite{davydychev,fljeta,pharie}.

As a first attempt at calculating the unordered integrals, one might try a brute-force approach
expanding all the exponential factors, and computing
\begin{gather*}
I_N({n_{12},n_{13},\ldots ,n_{(N-1)N}})
\equiv\int_{12\ldots N}\prod_{i<j=1}^N G_{ij}^{n_{ij}}.
\end{gather*}
However, although these integrals are individually trivial, it is by no means easy to arrive at a~closed-form
expression that eventually might allow one to perform a resummation. Instead, we~have found it more promising to
expand each exponential factor in terms of the inverse derivative operators above.

\section{Expansion in inverse derivatives}\label{section11}

Thus, we expand each exponential using the identity (whose proof is given in the appendix)
\begin{gather}
\e^{p_{ab} G_{ab}}=1+2 \sum_{n=1}^\infty p_{ab}^{n-1/2}H_{2n-1} \bigg(\frac{\sqrt{p_{ab}}}{2}\bigg)
\overline{\bra{u_a} \partial^{-2n} \ket{u_b}}.
\label{idcore}
\end{gather}
Here the $ H_n(x)$ are Hermite polynomials, and we have abbreviated
\begin{gather*}
\overline{\bra{u_a} \partial^{-2n} \ket{u_b}}\equiv
\bra{u_a} \partial^{-2n} \ket{u_b} - \bra{u_a} \partial^{-2n} \ket{u_a}.
\end{gather*}
By integration, this also gives
\begin{gather*}
\frac{\sqrt{\pi}}{2x}{\rm erf}(x) \e^{x^2} = 1+ \sum_{n=1}^\infty 2^{2n}\hat B_{2n} x^{2n-1}H_{2n-1}(x)
\end{gather*}
$\big(\hat B_n \equiv \frac{B_n}{n!}\big)$.
Curiously, we have not been able to find this simple series representation of the error function in the mathematical literature.

Let us have a look at the three-point case. Here using \eqref{idcore} in each factor leads to
\begin{gather*}
{\rm e}^{p_{12} G_{12}+p_{13} G_{13}+p_{23} G_{23}}=
\bigg\{1+2 \sum_{i=1}^\infty p_{12}^{{\rm i}-\frac{1}{2}}H_{2i-1} \bigg( \frac{\sqrt{p_{12}}}{2}\bigg)
\big[\bra{u_1} \partial^{-2i} \ket{u_2} + \hat{B}_{2i}\big] \bigg\}
\\ \hphantom{{\rm e}^{p_{12} G_{12}+p_{13} G_{13}+p_{23} G_{23}}=}
{}\times \bigg\{1+2 \sum_{j=1}^\infty p_{13}^{j-\frac{1}{2}}H_{2j-1}\bigg(\frac{\sqrt{p_{13}}}{2}\bigg) \big[\bra{u_1} \partial^{-2j}\ket{u_3} + \hat{B}_{2j} \big]\bigg\}
\\ \hphantom{{\rm e}^{p_{12} G_{12}+p_{13} G_{13}+p_{23} G_{23}}=}
{}\times\bigg\{ 1+2 \sum_{k=1}^\infty p_{23}^{k-\frac{1}{2}}H_{2k-1} \bigg( \frac{\sqrt{p_{23}}}{2}\bigg) \big[ \bra{u_2} \partial^{-2k} \ket{u_3} + \hat{B}_{2k}\big] \bigg\}.
\end{gather*}
Since $\int_0^1 {\rm d}u_i \bra{u_i} \partial^{-2n} \ket{u_j} = \int_0^1 {\rm d}u_j \bra{u_i} \partial^{-2n} \ket{u_j} = 0$, the three
$\bra{u_i} \partial^{-2n} \ket{u_j}$ must go
together, and then we can apply the completeness relation $\int_0^1du \ket{u} \bra{u} = \Eins$ to get
\begin{gather*}
\int_{123} \bra{u_1} \partial^{-2i}\ket{u_2} \bra{u_2} \partial^{-2k}\ket{u_3} \bra{u_3} \partial^{-2j}\ket{u_1} =\operatorname{Tr}\big( \partial^{-2(i+j+k)}\big) = -\hat{ B}_{2(i+j+k)},
\end{gather*}
where we have used \eqref{deln} in the last step. In~this way we get a closed form-expression for the $ N=3$ coefficients,
\begin{align*}
I_3(a,b,c)&\equiv
\int_{123} G_{12}^a G_{13}^b G_{23}^c
\\
&= a! b! c! \sum_{i=\lfloor 1+a/2 \rfloor}^a \sum_{j=\lfloor 1+b/2 \rfloor}^b \sum_{k=\lfloor 1+c/2 \rfloor}^c h_i^a h_j^b h_k^c\big( \hat{B}_{2i}\hat{B}_{2j}\hat{B}_{2k}- \hat{B}_{2(i+j+k)} \big).
\end{align*}
In writing this identity we have assumed that $a$, $b$, $c$ are all different from zero. The coefficients~$h_i^a$ can be found from the rearrangement
\begin{gather*}
2\sum_{i=1}^{\infty} \lambda^{i-\frac{1}{2}} H_{2i-1} \bigg(\frac{\sqrt{\lambda}}{2}\bigg)
{\hat B}_{2i}= \sum_{a=1}^{\infty} \lambda^a \sum_{i=\lfloor 1+a/2 \rfloor}^a h_i^a {\hat B}_{2i}.
\end{gather*}
From the explicit formula for the Hermite polynomials
\begin{gather*}
H_n(x) = n! \sum_{m=0}^{\lfloor{\frac{n}{2}\rfloor}}
(-1)^m\frac{(2x)^{n-2m}}{m!(n-2m)!}
\end{gather*}
we find
\begin{gather*}
h_i^a = (-1)^{a+1} \frac{2(2i-1)!}{(2i-a-1)!(2a-2i+1)!}.
\end{gather*}
Differently from the three-point case, for $N=4$ we encounter, apart from ``chain integrals'', the ``cubic worldline vertex''
\begin{gather*}
V_{3}^{{\rm i}jk} \equiv \int_0^1 {\rm d}u \bra{u} \partial^{-i}\ket{u_1} \bra{u} \partial^{-j}\ket{u_2} \bra{u} \partial^{-k}\ket{u_3} .
\end{gather*}
This vertex can be computed by partial integration as follows:
\begin{gather*}
V_{3}^{{\rm i}jk} \equiv \int_0^1 {\rm d}u \bra{u} \partial^{-i}\ket{u_1} \bra{u} \partial^{-j}\ket{u_2} \bra{u} \partial^{-k}\ket{u_3}
\\ \hphantom{V_{3}^{{\rm i}jk}}
{}=\sum_{a=i}^{{\rm i}+j-1}\!(-1)^a {{a-1}\choose{i-1}}\!\!
\int_0^1\!\! {\rm d}u \bra{u} \partial^{0}\ket{u_1} \bra{u} \partial^{-(i+j-a)}\ket{u_2} \bra{u} \partial^{-(k+a)}\ket{u_3}\!+\!\lbrace i\leftrightarrow j, u_1 \leftrightarrow u_2 \rbrace
\\ \hphantom{V_{3}^{{\rm i}jk}}
{}=\!\sum_{a=i}^{{\rm i}+j-1}\!(-1)^a {{a-1}\choose{i-1}}\big\lbrack \bra{u_1} \partial^{-(i+j-a)}\ket{u_2} \bra{u_1} \partial^{-(k+a)}\ket{u_3}\! -\! (-1)^{{\rm i}+j-a}\bra{u_2} \partial^{-(i+j+k)}\ket{u_3}\big\rbrack
\\ \hphantom{V_{3}^{{\rm i}jk}}
{} +\!\sum_{a=j}^{{\rm i}+j-1}\!(-1)^a {{a-1}\choose{j-1}}\big\lbrack \bra{u_2} \partial^{-(i+j-a)}\ket{u_1} \bra{u_2} \partial^{-(k+a)}\ket{u_3}
 \!-\! (-1)^{{\rm i}+j-a}\bra{u_1} \partial^{-(i+j+k)}\ket{u_3}\big\rbrack.
\end{gather*}
In the second step \eqref{invder} was used.
Here we are, incidentally, just reusing one of the algorithms that was used in the toy model described above to obtain
alternative representations for a~given multiple zeta value.
This algorithm generalizes to the worldline vertex $V_m^{i_1i_2\cdots i_m}$. In~the calculation of the $N$-point function, the worldline vertices $V_3,V_4,\ldots,V_{N-1}$ will be needed.

\section{Outlook}\label{section12}

Although the results that we have presented here on scalar $N$-point functions are still preliminary,
what is already clear is that the expansion \eqref{idcore} together with integration by parts allows one to
perform the unordered integrals, in principle for any $N$. The challenge then becomes to identify the
resulting multiple sums with known hypergeometric functions. As a byproduct, we will get
closed formulas for the basic coefficients $I_N\big({n_{12},n_{13},\ldots ,n_{(N-1)N}}\big)$. Those seem less relevant
for the momentum-space amplitudes, but fundamental for the associated heat-kernel expansion of the
effective action in $x$-space~\cite{6}.
Finally, the property of the Bernoulli numbers to converge rather fast towards their asymptotic limit,
\begin{gather*}
B_{2n} \sim (-1)^{n+1}2{(2n)!\over (2\pi)^{2n}}
\end{gather*}
might become useful for approximations.

\appendix

\section{Proof of identity (\ref{idcore})}\label{appendixA}

\noindent The following identity shall be proved.
\begin{Identity}
\begin{gather*}
 \e^{pG_{ab}}=1+2\sum_{n=1}^\infty p^{n-\frac{1}{2}}
 H_{2n-1}\bigg(\frac{\sqrt{p}}{2}\bigg)
 \overline{\left\langle u_a|\dell^{-2n}|u_b\right\rangle},
 \end{gather*}
 where $u_a$ and $u_b$ are $1$-periodic coordinates on a circle with
 length $1$, $G_{ab}=|u_a-u_b|-(u_a-u_b)^2$ $(|u_a|,|u_b|<1)$,
 $H_n(x)$ Hermite polynomials, $p$ some parameter, $\dell^{-1}$
 the inverse derivative on the circle with zero integration constant and
\begin{gather*}
\overline{\left\langle u_a|\dell^{-2n}|u_b\right\rangle}
= \left\langle u_a|\dell^{-2n}|u_b\right\rangle
-\left\langle u_a|\dell^{-2n}|u_a\right\rangle.
\end{gather*}
\end{Identity}
Using translation invariance on the circle, we can express the
identity in terms of the difference $x=u_a-u_b$. The inverse
derivatives can then be expressed by periodically continued Bernoulli
polynomials
\begin{gather*}
\left\langle u_a|\dell^{-n}|u_b\right\rangle =
\delta_{n,0}\sum_{k=-\infty}^\infty\delta(x-k) -\frac{1}{n!}\bar{B}_n(x),\qquad
n\ge 0,
\end{gather*}
where $\bar{B}_n(x)=B_n(x-[x])$ ($[x]$ denotes the largest integer
smaller or equal to $x$) coincides with the Bernoulli polynomial
$B_n(x)$ on the interval $x\in[0,1)$ and is periodically continued to
 $x\in\R$ by $\bar{B}_n(x+1)=\bar{B}_n(x)$ (hence, $\bar{B}_n(x)$ is
 no polynomial for $n>0$, when it is regarded as defined on
 $\R$). $\bar{B}_n(x)$ is continuous at $x\in\Z$ for all $n\neq 1$,
 whereas $\bar{B}_1(x)$ jumps from $\frac{1}{2}$ on the left to~$-\frac{1}{2}$ on the right of integer arguments, and the value at
 integers is chosen such that it equals the Bernoulli number $B_1$,
 which is $-\frac{1}{2}$ by convention.

 Using these definitions, and
 restricting $x=u_a-u_b$ to the interval $x\in[0,1)$, the identity
 can be written in terms of Bernoulli polynomials, which is the
 form that we are going to prove:
\begin{Identity}\label{identity-2}
\begin{gather*}
 \e^{px(1-x)}=1+2\sum_{n=1}^\infty p^{n-\frac{1}{2}}
 H_{2n-1}\bigg(\frac{\sqrt{p}}{2}\bigg)(B_{2n}(x)-B_{2n}(0)).
 \end{gather*}
\end{Identity}
The proof given below requires the following lemma:
\begin{Lemma}\label{lemma}
 Let $r$ be any odd positive integer, $x$ an indeterminate. Then the
 following polynomial identity holds
 \begin{equation*}
 \sum_{n=0}^r\binom{r}{n} (2x-r+n)_{n} (x)_{r-n} B_n=0,
 \end{equation*}
 where $B_n=B_n(0)$ denotes the $n$th Bernoulli number and
 $(a)_{0}=1$, $(a)_{n}=(a)_{n-1} (a-n+1)$ the Pochhammer symbol
 denoting here the falling factorial $($which is a polynomial if $a$ is one$)$.
\end{Lemma}
\begin{proof}
 For any odd $r$, the Bernoulli polynomial $B_r(x)$ has the property
 $B_r(1-x)=-B_r(x)$. Hence, given any integrable function
 $f\colon [0,1]\to\C$ with $f(1-x)=f(x)$, we have by substituting $x\mapsto
 1-x$ and using the transformation properties
\begin{gather*}
 \int_0^1f(x)B_r(x)\,\dd x = \int_0^1f(1-x)B_r(1-x)\,\dd
 x=-\int_0^1f(x)B_r(x)\,\dd x=0.
 \end{gather*}
 Applying this to the explicit form of the Bernoulli polynomials,
\begin{gather}\label{Bernoulli-polynomial-explicit}
 B_r(x)=\sum_{n=0}^r\binom{r}{n}B_nx^{r-n},
 \end{gather}
 we get
\begin{gather*}
 \sum_{n=0}^r\binom{r}{n}B_n\int_0^1f(x) x^{r-n}\dd
 x=0,\qquad f(x)=f(1-x),\quad r>0\ \text{odd}.
 \end{gather*}
 Choosing $f(x)=x^{u-1}(1-x)^{u-1}$ for any $u\in\C$, $\operatorname{Re} u>0$, we can
 express the integral explicitly by $\Gamma$-functions
 \begin{gather*}
 \int_0^1x^{u-1}(1-x)^{u-1}x^{r-n}\dd
 x=\frac{\Gamma(u+r-n)\Gamma(u)}{\Gamma(2u+r-n)}
 \end{gather*}
 and obtain an identity of meromorphic functions in $u$
\begin{gather*}
 \sum_{n=0}^r\binom{r}{n}B_n\frac{\Gamma(u+r-n)\Gamma(u)}{\Gamma(2u+r-n)} =0,\qquad
\mbox{$r>0$ odd}.
 \end{gather*}
 Using the $\Gamma$-function identity $\Gamma(x+1)=x \Gamma(x)$ and
 assuming $r,n\in\Z$, $0\le n\le r$, we can introduce
 Pochhammer symbols\footnote{Other Pochhammer representations using
 also $\Gamma(x)\Gamma(1-x)=\frac{\pi}{\sin(\pi x)}$ do not give
 other identities, because they can also be obtained by the
 identity $(a)_{n}=(-1)^n(-a+n-1)_{n}$, where each factor is multiplied
 with $-1$.} for falling factorials
\begin{gather*}
 \Gamma(u+r-n)=(u+r-n-1)_{r-n}\Gamma(u) =(-1)^{r-n}(-u)_{r-n}\Gamma(u),
 \\
 \Gamma(2u+r-n)=\frac{\Gamma(2u+r)}{(2u+r-1)_{n}} =(-1)^n\frac{\Gamma(2u+r)}{(-2u-r+n)_{n}}
 \end{gather*}
 and obtain after dividing out the factors independent of $n$
\begin{gather*}
 \sum_{n=0}^r\binom{r}{ n} B_n (-u)_{r-n} (-2u-r+n)_{n} =0,\qquad \mbox{$r>0$ odd}.
 \end{gather*}
 Now we have an equality between analytic functions in~$u$, which are
 in fact polynomials, such that we can replace the complex variable
 $-u$ with the indeterminate $x$ and obtain the polynomial identity
 of the lemma.
\end{proof}
\begin{Corollary}\label{corollary}
\begin{gather*}
\sum_{n=0}^{2r+1}\binom{2r+1}{ n} B_n \binom{n+2m-2r-1}{ m}=0,\qquad
r,m\in\Z,\quad 0\le r<m.
\end{gather*}
\end{Corollary}
\begin{proof}
 Writing the second binomial coefficient as
\begin{align*}
\binom{n+2m-2r-1}{m}&=\frac{1}{m!}(n+2m-2r-1)_{m}
\\
&=\frac{1}{m!}(n+2m-2r-1)_{n}(2m-2r-1)_{m-2r-1}(m)_{2r+1-n},
\end{align*}
the sum of the corollary becomes
\begin{gather*}
 \frac{1}{m!}(2m-2r-1)_{m-2r-1}\sum_{n=0}^{2r+1}\binom{2r+1}{n} B_n (n+2m-2r-1)_{n}(m)_{2r+1-n},
 \end{gather*}
 which is zero by Lemma \ref{lemma} setting $x=m$.
\end{proof}
We come now to the main proof.
\begin{proof}[Proof of Identity~\ref{identity-2}]
 Using the explicit expression for the Hermite polynomials
\begin{gather*}
 H_n(x)=\sum_{k=0}^{\left[\frac{n}{2}\right]}\frac{n!}{(n-2k)! k!} (-1)^k (2x)^{n-2k},
 \end{gather*}
 and comparing equal powers $p^m$ ($m\ge 1$), the identity reduces to
 the equivalent form
\begin{gather*}
 \frac{1}{m!}x^m(1-x)^m=
 \sum_{\frac{m+1}{2}\le n\le m}\frac{2}{(2n-m-1)! (2m-2n+1)!}
 \frac{1}{2n}(B_{2n}(x)-B_{2n}(0)).
 \end{gather*}
 Multiplying by $m!\neq0$, inserting the explicit form
 (\ref{Bernoulli-polynomial-explicit}) for the Bernoulli polynomials
 in terms of Bernoulli numbers, expanding the binomial $(1-x)^m$ and
 comparing equal powers $x^k$ ($k\ge1$), the identity takes the
 equivalent form
\begin{gather*}
 \sum_{\frac{k}{2}\le n\le m}\frac{1}{2n}\binom{2n}{k}
 2 B_{2n-k} \binom{m}{ 2n-m-1}=(-1)^k \binom{m}{ k-m},\qquad 1\le k\le 2m,
 \end{gather*}
 or equivalently, replacing $2n$ by a new summation variable $n$
 restricted to even values
\begin{gather}\label{condition}
 \sum_{\substack{n=k\\\mbox{\scriptsize $n$ even}}}^{2m}\frac{1}{n}\binom{n}{
 k} 2 B_{n-k} \binom{m}{ n-m-1}=(-1)^k \binom{m}{ k-m},\qquad 1\le k\le 2m.
 \end{gather}
 Now we have two cases:

1. $k$ is odd. Then, because all odd Bernoulli numbers, except
 $B_1=-\frac{1}{2}$, vanish, the sum in~(\ref{condition})
 consists of the single term for $n=k+1$, which gives
\begin{gather*}
 \frac{1}{k+1}\binom{k+1}{ k} 2 B_1 \binom{m}{k-m}=-\binom{m}{ k-m},
 \end{gather*}
 which is just the right-hand side for odd $k$. Hence,
 (\ref{condition}) is true in this case.

2. $k$ is even. Then the right-hand side of (\ref{condition}),
 having a positive sign, fits into the sum as additional term for
 $n=k+1$. Thus, in this case, taking into account that all other odd
 Bernoulli numbers vanish, equation (\ref{condition}) is equivalent  to
\begin{gather*}
 \sum_{n=k}^{2m+1}\frac{1}{n}\binom{n}{k} 2 B_{n-k} \binom{m}{n-m-1}=0,
 \end{gather*}
 where the summation extends now over all, even and odd, integers,
 and the range had to be extended to $n=2m+1$ to cover also the case
 $k=2m$, where the additional term for $n=k+1$ is exactly the term
 for $n=2m+1$.

 Dividing by $2$, shifting $n$ by $k$, writing the binomial
 coefficients in factorials and introducing $r\ge0$ by setting $k=2m-2r$,
 we obtain
\begin{gather*}
 \sum_{n=0}^{2r+1}\frac{(n+2m-2r-1)!}{n! (2m-2r)!}
 B_{n} \frac{m!}{(2r+1-n)! (n+m-2r-1)!}=0.
 \end{gather*}
 Multiplying this equation with $\frac{(2r+1)!(2m-2r)!}{(m!)^2}\neq
 0$ and recombining the factorials into binomial coefficients, it can
 be written equivalently as
\begin{gather*}
 \sum_{n=0}^{2r+1}\binom{2r+1}{ n} B_{n} \binom{n+2m-2r-1}{ m}=0,
 \end{gather*}
 which is true by Corollary~\ref{corollary}. Hence, (\ref{condition})
 is true also in this case.
\end{proof}

\subsection*{Acknowledgements}
We thank Andrei Davydychev and Tord Riemann for sharing with us their expertise on scalar off-shell $N$-point functions.

\pdfbookmark[1]{References}{ref}
\LastPageEnding

\end{document}